\documentclass[conference]{IEEEtran}
\usepackage{newtxtext}
\usepackage{placeins}

\usepackage{amsmath,amssymb,amsthm}
\usepackage{graphicx}
\usepackage{booktabs}
\usepackage{multirow}
\usepackage{xcolor}
\usepackage{cite}
\usepackage[hidelinks]{hyperref}
\usepackage{tikz}

\newtheorem{proposition}{Proposition}

\newcommand{\sys}{AURA}
\newcommand{\Asl}{\ensuremath{A_{\mathrm{sl}}}}
\newcommand{\Aen}{\ensuremath{A_{\mathrm{en}}}}

\graphicspath{{paper-artifacts-v2/}}

\begin{document}

\title{Taming the Agentic RAN: Stability-Guaranteed Arbitration of\\ Autonomous AI Agents in O-RAN}
\author{
  \IEEEauthorblockN{Seyed Bagher Hashemi Natanzi, Bo Tang}
  \IEEEauthorblockA{
    Department of Electrical and Computer Engineering \\
    Worcester Polytechnic Institute, Worcester, MA, USA \\
    \{snatanzi, btang1\}@wpi.edu
  }
}

\maketitle
    \begin{tikzpicture}[remember picture, overlay]
        \node[anchor=north, yshift=-0.5cm] at (current page.north) {\fbox{\parbox{\textwidth}{\centering\small {\color{red}This work has been submitted to the IEEE Conference for possible publication. Copyright may be transferred without notice, after which this version may no longer be accessible.}}}};
    \end{tikzpicture}
\begin{abstract}
The O-RAN control plane is becoming agentic: autonomous AI agents, deployed as rApps by different vendors, independently close control loops over shared radio resources. We demonstrate on a live O-RAN system that this independence is unsafe. Two agents with individually correct objectives, one protecting a latency SLA and one maximizing utilization for energy efficiency, jointly drive recurring opposing excursions of the shared resource partition that neither produces alone. Existing conflict-mitigation mechanisms presume a statically known application population and cannot govern agents whose behavior emerges at run time. We present \sys{}, a lightweight arbitration layer that admits agent actions only when they satisfy feasibility invariants, per-variable dwell times, and a deadband, and we prove the arbitrated system converges to a feasible operating point. Implemented on an OpenAirInterface (OAI) testbed with measured one-way latency and throughput, \sys{} reduces recurring shared-state excursions by more than an order of magnitude (from 8.4 to 0.4~PRB amplitude) and virtually eliminates cross-slice throughput starvation (from 40--55\% to 0.3\%), while leaving the protected slice's own latency compliance unchanged, a trade-off the convergence guarantee makes explicit.
\end{abstract}


\section{Introduction}
 
The future AI RAN control plane is agentic. Rather than executing fixed, hand-tuned policies, AI agents observe network telemetry, reason about operator-level intents, and act autonomously on the network without per-decision human oversight. An operator's RAN will host not one such agent but many, developed by different vendors for different purposes: one rApp managing slice quotas to protect latency SLAs, another shrinking energy budgets when capacity goes unused, a third steering traffic to balance load across the network. The O-RAN architecture enables this ecosystem through near-real-time 
xApps on E2 for sub-second control and non-real-time rApps on A1 and 
O1 for longer-horizon policy and orchestration~\cite{polese2023understanding}.
 
What the architecture does not provide is any guarantee about what happens when these agents run \emph{together}. We show that the answer can be dangerous. Two agents with individually reasonable objectives jointly drive repeated opposing excursions of the shared radio resource partition, a pathology that is absent when either agent runs alone and that recurs across independent load cycles. The mechanism is structural, not accidental. An SLA-protection agent enlarges its slice quota to buy latency headroom. An energy agent interprets the resulting spare capacity as waste and shrinks the total resource budget. The smaller budget re-inflates latency, the SLA agent responds, and the loop repeats on every disturbance. Neither agent observes the other, since the O-RAN architecture provides no channel through which independently deployed applications learn of each other's objectives. Their interaction therefore forms a delayed feedback loop between opposing controllers, a well-known precondition for instability that, to our knowledge, has not previously been demonstrated end-to-end on a running O-RAN system.
 
Conflict between RAN applications is a recognized concern, and existing work addresses it for xApps with statically known parameter footprints~\cite{delrio2025pacifista,zolghadr2025conflict,adamczyk2023conflict}. The agentic setting, however, is qualitatively different. An xApp has fixed logic and a known action space; its conflicts with other xApps can be identified offline through static pairwise analysis. An AI agent is designed to be autonomous: it may use online learning so that its behavior evolves with network conditions, it may carry internal memory that shapes future decisions, and its architecture (DRL, LLM-backed, or otherwise) may be entirely opaque to the operator and to other agents. This means that agentic conflicts are not a design-time artifact but a runtime phenomenon. They emerge from the interaction of behaviors that are individually correct, and they cannot be predicted by analyzing any single agent in isolation. No existing mechanism addresses this regime.
 
We present \sys{} (Arbitrated aUtonomous Resource Agents), a lightweight coordination layer that fills this gap. \sys{} requires no knowledge of how any agent decides. Instead, agents submit their intended actions as \emph{proposals} that declare their effect on shared network state. A thin arbiter enforces three conditions before any proposal reaches the network: the resulting state must satisfy system feasibility invariants, the affected variable must not have been modified too recently, and the change must exceed a meaningful threshold. When proposals conflict, a priority ordering resolves them in favor of SLA-critical actions. We prove that this design converges to a feasible operating point and validate it on a live OpenAirInterface/FlexRIC O-RAN testbed using real measured traffic.
 
Our contributions are: (1) the first end-to-end empirical demonstration of multi-agent RAN instability on a running 5GSA O-RAN stack, using real measured one-way delay and throughput rather than model-derived proxies; (2) a formal model of the pathology as a delayed opposing best-response process, with a limit-cycle proof and a convergence guarantee for the arbitrated system; (3) an implementation of per-slice PRB quota enforcement in the OAI NR MAC downlink pre-processor, which is absent in upstream OAI, together with a measurement of the full control-loop latency budget from agent decision to MAC enforcement. Arbitration damps recurring shared-state excursions by an order of magnitude and virtually eliminates the throughput starvation that unarbitrated agents inflict on coexisting slices, at the explicitly acknowledged cost of not improving the protected slice's own latency compliance, a trade-off our convergence guarantee makes honest rather than surprising.
 
The lesson extends beyond two agents and one cell. Whenever autonomous controllers with contradictory objectives share delayed-observation loops over a common resource, coordination must be enforced at the boundary where actions meet the network. \sys{} shows that a provably correct version of that boundary costs almost nothing to build.
 
\section{Related Work}

\textbf{RAN slicing and its enforcement.} Network slicing partitions one physical network into logical networks with distinct service guarantees; its RAN realization requires the MAC scheduler to partition radio resources among slices~\cite{foukas2017orion,ksentini2017slicing}. NVS~\cite{kokku2012nvs} introduced rate- and capacity-based slice scheduling; subsequent systems brought slicing to programmable stacks~\cite{foukas2016flexran,schmidt2021flexric}. On OAI specifically, ORANSlice~\cite{cheng2024oranslice} implements NSSAI-based PRB allocation with min/max/dedicated ratios, and xSlice~\cite{xslice2025} rewrites the downlink scheduler for DRL-driven slice resizing. Our enforcement mechanism follows the \emph{quota} style of this line of work rather than positional PRB ranges; our contribution is not the slicer but what happens when multiple autonomous controllers manipulate it concurrently. We additionally document a practical gap relevant to reproducibility: in current upstream OAI, the E2 slice service model returns emulated indication data and acknowledges control without acting on the MAC, so any multi-agent study on this path requires real enforcement first (Section~\ref{sec:impl}).

\textbf{Conflict mitigation in O-RAN.} Beyond enforcement mechanisms, conflict mitigation is essential. The O-RAN architecture names conflict mitigation as a Near-RT RIC function~\cite{oranwg3}. PACIFISTA~\cite{delrio2025pacifista} profiles xApps offline to derive conflict likelihoods; graph-learning approaches infer conflicts among slicing xApps~\cite{zolghadr2025conflict}; other work classifies direct, indirect, and implicit conflicts and proposes controller-side resolution~\cite{adamczyk2023conflict}. These mechanisms assume a fixed population of applications whose tunable parameters are known in advance. \sys{} targets the complementary regime: agents instantiated dynamically, with opaque decision logic, whose conflicts emerge only at run time. Unlike prior work, our arbiter inspects \emph{declared effects on shared state} rather than application internals, and it carries a convergence guarantee validated on a live stack.

\textbf{Multi-agent control and stability.} That independent controllers with coupled objectives can oscillate is classical: best-response dynamics need not converge without structure such as potential games~\cite{monderer1996potential}, and delayed feedback destabilizes otherwise stable loops. Dwell-time conditions are a standard tool for stabilizing switched systems~\cite{liberzon2003switching}. Our theoretical contribution is not new mathematics but the mapping of these tools onto the O-RAN control plane. We identify the shared-variable structure, the delay sources, and an arbitration rule whose admitted-action sequence provably terminates, together with empirical validation that the pathology and its cure both manifest on a real 5G system under real measured traffic.

\textbf{Agentic and LLM-driven network control.} Recent work explores LLM-based agents for network management and intent translation~\cite{wu2024netllm}. This trend sharpens our motivating assumption: future control applications will be numerous, heterogeneous, and unanalyzable in advance. \sys{} treats the agent as a black box by design; the proposal schema is identical for rule-based and learning-based agents, and our own agents are implemented as simple rule-based controllers precisely so that the observed instability and its repair can be attributed to the multi-agent \emph{interaction structure} rather than to the sophistication of any one decision policy.

\section{System Model}\label{sec:model}
 
\textbf{Resources and slices.} We consider a single cell with $P$ physical resource blocks (PRBs) available per downlink slot. The cell serves $K$ slices; slice $i$ holds a quota $q_i \in [0,P]$, the maximum PRBs its associated users may receive per slot, enforced in the MAC downlink pre-processor. A cap $C \le P$ bounds the total allocatable budget, $\sum_i q_i \le C$; the enforcement mechanism (Section~\ref{sec:impl}) realizes this bound by scaling slice quotas down whenever their sum exceeds $C$, modeling an energy-saving posture in which capacity beyond the cap is not allocated. The shared control state is $x = (q_1,\dots,q_K,C)$. We deliberately use \emph{quota} semantics rather than positional PRB intervals $[p_{\mathrm{low}},p_{\mathrm{high}}]$: positional partitions interact with control-channel and broadcast allocations, while quotas compose with the existing proportional-fair allocator and match the shared variables of our theory.
 
\textbf{Traffic and SLAs.} Slice~1 carries latency-sensitive traffic with target $\ell_1(t) \le L_{\max}$, where $\ell_1(t)$ is measured one-way delay over a sliding window from a UDP probe with a timestamp in the payload; slice~2 carries throughput-oriented traffic with target $r_2(t) \ge R_{\min}$, measured at the receiver. Both SLA metrics are monotone in effective service rate: for fixed offered load, $\ell_1$ is non-increasing in $q_1$ and $r_2$ is non-decreasing in $q_2$, saturating at the offered load. Utilization is $U(t) = \frac{1}{C}\sum_i u_i(t)$ with $u_i(t)$ the granted PRBs per slot averaged over the window; $U$ is defined \emph{relative to the cap}, so tightening $C$ raises measured utilization without serving additional traffic.
 
\textbf{Agents.} In the O-RAN architecture, our agents are realized as \emph{rApps} running in the non-real-time RIC: they receive telemetry through O1/R1 interfaces and, in a production deployment, propose actions via A1 policies enforced through the non-RT-to-near-RT path; our prototype drives the same shared quota table through a file-based mechanism (section \ref{sec:impl}). This placement is appropriate for objectives that operate on multi-second horizons (slice quota and energy cap), while leaving sub-second scheduling to the near-RT RIC. Two autonomous agents observe the system through periodic telemetry with period $T_o$ and act with control period $T_a \ge T_o$. The SLA agent \Asl{} observes $(\ell_1, r_2)$: if $\ell_1 > L_{\max}$ it proposes $q_1 \mathrel{+}= \delta$; if $r_2 < R_{\min}$ it proposes $q_2 \mathrel{+}= \delta$; when both SLAs hold with margin it decays quotas toward a nominal point. The efficiency agent \Aen{} observes $U$: if $U < U_{\mathrm{tgt}}$ it proposes $C \mathrel{-}= \delta_c$; if $U > U_{\mathrm{hi}}$ it proposes $C \mathrel{+}= \delta_c$. Each agent is individually rational: run alone, \Asl{} moves only its own quota and never touches $C$ (validated in \S\ref{sec:eval}). Neither observes the other's objectives, proposals, or existence; the O-RAN architecture provides no such channel between independently deployed applications.

 Our agents are rule-based in this study by design, not by architectural
necessity: the arbiter's admission logic operates only on the declared
proposal $(a, \Delta x, \rho, \mathrm{TTL})$ and is agnostic to whether
that proposal originates from a fixed policy, a fine-tuned DRL model, or
an LLM-backed planner, a property we exploit to attribute the observed
instability to agent \emph{interaction} rather than to the sophistication
of either agent's decision logic (see section \ref{sec:discussion}).

\textbf{Observation delay.} A decision at time $t$ reflects telemetry from the interval $[t - T_o - d,\; t - d]$, where $d$ aggregates measurement, transport, and enforcement latency. We write $\tau_{\mathrm{loop}} = T_o + d$ for the loop delay (in our deployment $T_o = T_a$); \S\ref{sec:eval} reports its measured distribution.
 
\textbf{Arbitration.} In the \emph{direct} regime, agent actions are enforced immediately. In the \emph{arbitrated} regime, each action is a proposal $\pi = (a, \Delta x, \rho, \mathrm{TTL})$ from agent $a$ carrying its declared effect $\Delta x$ on shared state and a priority class $\rho$. The arbiter admits $\pi$ iff (i)~the post-state satisfies the invariants $\sum_i q_i \le C \le P$ and $q_1 \ge q_1^{\mathrm{floor}}$; (ii)~no variable touched by $\Delta x$ was modified within the past $\tau_d$ seconds (dwell); and (iii)~$\lVert \Delta x \rVert \ge \epsilon$ (deadband). Conflicting proposals within an admission window resolve lexicographically by $\rho$: SLA-restoring $\succ$ throughput $\succ$ efficiency. Rejected proposals are returned to the agent with the violated condition; agents may re-propose after re-observing.
\section{Instability and Its Removal}\label{sec:theory}

\subsection{Why the agents fight}

\Asl{} reduces $\ell_1$ by increasing the headroom $C - \sum_i u_i$; \Aen{} eliminates exactly that headroom. Define the best-response maps
\begin{align}
B_{\mathrm{sl}}(C) &= \min\{\, q_1 : \ell_1(q_1; C) \le L_{\max} \,\}, \\
B_{\mathrm{en}}(q_1) &= \max\{\, C : U(q_1, C) \ge U_{\mathrm{tgt}} \,\}.
\end{align}
Under saturating background load, $B_{\mathrm{sl}}$ is decreasing in $C$ (a tighter cap forces a larger quota to hold latency) and $B_{\mathrm{en}}$ is decreasing in $q_1$ (a larger protected quota lowers measured utilization, driving the cap down). Two decreasing maps compose to an increasing map; let $g = \lvert B_{\mathrm{sl}}' \cdot B_{\mathrm{en}}' \rvert$ denote its gain near the fixed point $x^{\ast} = (q_1^{\ast}, C^{\ast})$ where both constraints bind.

\subsection{Uncoordinated dynamics oscillate}

\begin{proposition}[Limit cycle under delay]\label{prop:cycle}
Suppose $g > 1$ in a neighborhood of $x^{\ast}$ and both agents act every $T_a$ on observations delayed by $\tau_{\mathrm{loop}} \ge T_a$, taking steps of size $\delta$ (resp.\ $\delta_c$). Then $x^{\ast}$ is unstable under the direct regime, and trajectories enter a limit cycle whose amplitude is bounded below by $\delta \cdot \lceil \tau_{\mathrm{loop}} / T_a \rceil$.
\end{proposition}

\begin{IEEEproof}[Proof sketch]
With stale observations, an agent continues stepping in the same direction for the $\lceil \tau_{\mathrm{loop}}/T_a \rceil$ periods during which its own most recent actions are not yet visible in telemetry, overshooting its best response by at least $\delta \lceil \tau_{\mathrm{loop}}/T_a \rceil$. Composition with gain $g>1$ amplifies the overshoot each round trip, so no trajectory remains within any sufficiently small neighborhood of $x^{\ast}$; the quota floor and cap ceiling confine trajectories to a bounded invariant set. On the quantized state lattice induced by step sizes $(\delta, \delta_c)$, the dynamics within this set are a deterministic map on finitely many states, hence eventually periodic, and by instability of $x^{\ast}$ the period exceeds one: a limit cycle. A complete argument is deferred to an extended version.
\end{IEEEproof}

The proposition predicts that the excursion should recur under repeated disturbance rather than settle after one occurrence, and that it should vanish when either agent is disabled. Our two-cycle experiment is designed to test the first prediction directly: under \textsc{Direct}, the shared cap moves again during the second, independent load cycle in a majority of repetitions (3 of 5, amplitudes 0, 2, 0, 2, 6 PRBs; Section~\ref{sec:eval}), and the amplitude is zero whenever either agent runs alone (Table~\ref{tab:amp}). Systematic step-size and delay sweeps are left to future work.

\subsection{Arbitration restores convergence}

\begin{proposition}[Arbitrated convergence]\label{prop:conv}
Under arbitration with dwell time $\tau_d > \tau_{\mathrm{loop}}$ and deadband $0 < \epsilon \le \min(\delta, \delta_c)$, so that agent steps remain admissible while zero-effect proposals are excluded, every admitted action is computed from an observation reflecting all previously admitted actions on the variables it touches. The admitted-action sequence is then a sequential best-response process; with the lexicographic priority order over a finite quantized state space, it satisfies the finite improvement property and terminates in finitely many admitted actions at a state within $\epsilon$ of a feasible point satisfying all invariants, with $q_1 \ge q_1^{\mathrm{floor}}$ maintained throughout by admission control.
\end{proposition}

\begin{IEEEproof}[Proof sketch]
Dwell $\tau_d > \tau_{\mathrm{loop}}$ guarantees that between consecutive admitted modifications of any variable, at least one full telemetry cycle elapses, so the stale-observation overshoot driving Prop.~\ref{prop:cycle} cannot occur; actions on each variable are serialized against fresh state. Order admitted actions by the lexicographic priority vector: each admitted non-redundant action (deadband excludes redundant ones) strictly improves the highest-priority violated objective without worsening any higher-priority satisfied one, since such a proposal would violate an invariant and be rejected. The priority vector over the finite lattice therefore admits no infinite strictly improving sequence, giving termination; invariants hold at every step because admission checks the post-state.
\end{IEEEproof}

Prop.~\ref{prop:conv} is deliberately modest: it does not claim optimality of the terminal point, only feasibility, SLA-floor safety at every step, and the absence of the recurring excursion. This matches the arbiter's design philosophy of constraining \emph{what} agents do to shared state while remaining agnostic to \emph{how well} they decide, and it correctly predicts the trade-off we measure in Section~\ref{sec:eval}: the arbiter damps the shared state and protects the coexisting slice, but it does not, and does not claim to, improve the protected slice's own SLA compliance beyond what the pre-existing allocation already delivers.

\section{\sys{} Design and Implementation}\label{sec:impl}

\subsection{Architecture}

Fig.~\ref{fig:arch} shows the system. Agents run in the non-real-time tier of our orchestration layer\footnote{Platform name withheld for double-blind review.} above the RIC; the arbiter interposes on the single path by which agent actions reach the network. Enforcement reaches the OAI gNB, where a quota module in the NR MAC downlink pre-processor applies the partition. Telemetry returns through the E2 MAC service model into a time-series store that both agents and our measurement harness read. For actuator validation we drive the quota table from a file that the pre-processor reloads periodically (isolating MAC enforcement from the E2 path); agent experiments then use the same table, written by the arbiter through the FlexRIC slice control API once the E2 callbacks are bound to it. All events, comprising observations, proposals, admissions, rejections, and enforcements, are logged to a single timeline keyed by one host clock, which the RF-simulated deployment makes exact.

\begin{figure*}[t]
\centering
\includegraphics[width=1\linewidth]{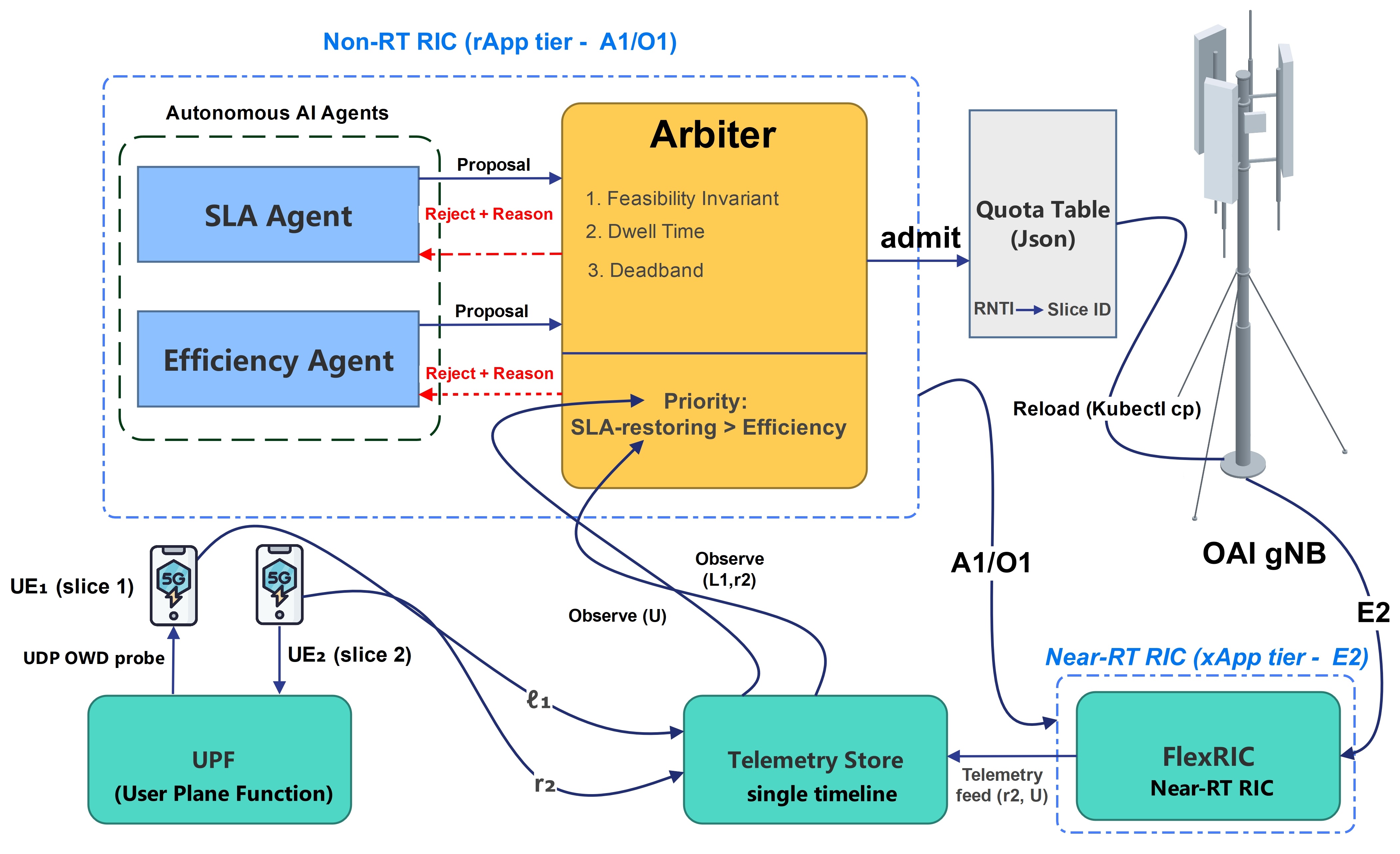}
\caption{\sys{} architecture. Autonomous AI agents ($A_\mathrm{sl}$, $A_\mathrm{en}$) in the Non-RT RIC submit proposals to a central arbiter, which enforces three conditions (feasibility invariant, dwell time, deadband) with SLA-restoring actions taking priority. Admitted proposals update a shared quota table reloaded by the OAI gNB; the NR MAC downlink pre-processor enforces per-slice PRB quotas. One-way latency ($\ell_1$) is measured via a dedicated UDP probe to UE$_1$; slice-2 throughput ($r_2$) is read at UE$_2$; MAC statistics return via E2 to FlexRIC and feed a shared telemetry store read by both agents.}
\label{fig:arch}
\end{figure*}

\subsection{PRB quota enforcement in the OAI NR MAC}

We report a finding of independent interest. In current upstream OAI, the E2 slice service model is an emulator: indication messages are filled with random data and control messages are acknowledged without touching the MAC, and upstream documentation states the model is supported only against emulated agents. Studies that exercise this path without verifying enforcement measure nothing. We therefore implement enforcement directly: a quota table mapping slice identifier to maximum PRBs per slot and UE (by RNTI) to slice, consulted by the downlink pre-processor, which caps each UE's allocation at its slice's remaining budget for the slot; the cap $C$ is enforced by scaling every slice's quota down whenever their sum exceeds $C$, so $C$ has a direct physical effect on delivered PRBs rather than serving only as arbiter-internal bookkeeping. We chose quota semantics over positional PRB ranges deliberately: positional partitions interact with control-channel and broadcast allocations, while quotas compose cleanly with the existing proportional-fair allocator and match the semantics of prior OAI slicers~\cite{cheng2024oranslice}. After the MAC path is validated, the slice service model's read and write callbacks are rewired to the same table, replacing the emulator, so standard xApps observe and control the real partition. The patch will be released with the paper.

\subsection{The arbiter}

The arbiter is a small service that maintains only the current shared state, per-variable last-modification timestamps, and the invariant set. Proposals arrive over a local API; admission executes the three checks of Section~\ref{sec:model}; admitted actions update the quota table and the timeline log. Two properties matter for deployment. First, the arbiter needs no model of agent internals, since rule-based and learning-based agents are indistinguishable behind the proposal schema. Second, rejection is informative: the violated condition is returned, so a well-behaved agent re-observes rather than blindly re-proposes.

\subsection{Measuring what the agents observe}\label{sec:owd}

An earlier prototype drove \Asl{} from a latency \emph{proxy} computed from the current quota and offered-load fill ratio rather than from a measured quantity. That proxy is adequate for exercising the arbiter's logic in isolation, but it cannot support any empirical claim about SLA compliance, since its value is deterministic given the shared state and therefore identical across repeated runs. All results reported in Section~\ref{sec:eval} instead use a dedicated UDP one-way-delay probe: a lightweight sender on the core's user-plane function embeds a timestamp in each payload at a fixed rate, and a receiver on the SLA-protected UE computes one-way delay against the shared host clock of the RF-simulated deployment, which makes the measurement exact rather than an RTT approximation. Throughput for both slices is read from receiver-side counters on the UEs. We verified the bidirectional reachability of this path with an independent echo test before every experimental run (Section~\ref{sec:actuator-gate}).

\section{Evaluation}\label{sec:eval}

\subsection{Setup}

Experiments run on a containerized 5G system: OAI core, one OAI gNB in RF-simulator mode (band n78, 20\,MHz, $P{=}51$ PRBs at 30\,kHz SCS), two OAI UEs, and FlexRIC as near-RT RIC, orchestrated on Kubernetes on a single server (AMD EPYC 9354, $2\times32$ cores / 64 threads, 377\,GiB RAM); gNB and UE processes are core-pinned (gNB to cores 0--7, UE$_1$ to 8--11, UE$_2$ to 12--15) and per-phase CPU headroom is reported in the appendix to exclude host scheduling as a confound. UE$_1$ (slice~1) is the SLA-protected slice and its one-way delay $\ell_1$ is measured by the probe of Section~\ref{sec:owd}. UE$_2$ (slice~2) receives UDP downlink measured at the receiver. Throughput experiments use UDP rather than TCP so that MAC quota effects are not confounded by congestion control or GTP path-MTU artifacts on the PDU session. SLA targets are $L_{\max}{=}18$\,ms and $R_{\min}{=}8$\,Mbps; agent parameters $T_a{=}5$\,s, $\delta{=}2$ PRBs (slice) / $\delta_c{=}2$ PRBs (energy cap), $U_{\mathrm{tgt}}{=}0.72$, $U_{\mathrm{hi}}{=}0.92$; arbiter parameters $\tau_d{=}8$\,s (chosen, per Prop.~\ref{prop:conv}, above a measured and estimated $\tau_{\mathrm{loop}} \approx 7$\,s dominated by enforcement latency, see Section~\ref{sec:latency-budget}), $\epsilon{=}1$ PRB, with initial state $q_1{=}22$, $q_2{=}18$, $C{=}48$ (of $P{=}51$). Each experiment applies a traffic script with two independent load cycles: 60\,s of steady state at 15\,Mbps offered per UE, then twice in succession a 90\,s step increasing UE$_1$'s offered load to 22\,Mbps (UE$_2$ unchanged) followed by a 60\,s step-down back to 15\,Mbps, for a total run length of 360\,s. The script is repeated 5 times per regime (20 runs total), with per-UE offered loads chosen within the region validated in Section~\ref{sec:actuator-gate} so that the agent dynamics are not confounded by host-side dual-UE compute contention.

We compare four regimes: \textsc{Static} (fixed quotas, no agents), \textsc{Single} (only \Asl{} active), \textsc{Direct} (both agents, no arbiter), and \textsc{\sys{}} (both agents, arbitrated).

\subsection{Actuator validation}\label{sec:actuator-gate}

Before any agent experiment we validate that the quota table actually moves delivered throughput, rather than trusting the emulated E2 acknowledgment (Section~\ref{sec:impl}), and that the UE-facing user-plane path is reachable in both directions so that no measurement silently reads a dead path. Control for the quota tests in this subsection is file-driven directly into the pre-processor, isolating MAC enforcement from the E2 path, so a positive result cannot be attributed to the stub control path.

\textbf{Single-UE sweep (primary actuator evidence).} Under saturating UDP downlink (offered 80\,Mbps) to one UE, raising its quota from 13 to 38 PRBs (of $P{=}51$) raises measured receiver goodput from 17.0 to 43.1\,Mbps, a $2.54\times$ increase consistent with the quota ratio, while the MAC service model's per-slot allocation counter tracks the configured quota exactly (a hard cap: quota$=13 \Rightarrow$ allocated $\approx 13$ PRBs/slot). This isolates the actuator from any dual-UE scheduling interaction and is our primary evidence that the quota mechanism, and not an artifact of aggregate offered load, controls delivered throughput.

\textbf{Two-UE fairness and directionality.} At moderate offered load, two UEs in distinct slices split throughput fairly under equal quotas (23.1/23.1\,Mbps at 22\,Mbps offered per UE), and an asymmetric 13:38 quota split at 28\,Mbps offered shifts throughput directionally (19.1/29.4\,Mbps). At offered loads approaching the cell's dual-UE ceiling (${\gtrsim}$28--40\,Mbps per UE, and under some asymmetric splits at lower offers), however, one UE can starve under RFsim's shared software-PHY compute load rather than under the quota mechanism itself. During affected runs we observed elevated, though sub-saturating (64--71\%), CPU utilization on the L1 transmit and radio-unit threads, and the starvation persists under explicit core pinning. This confirms a genuine dual-softmodem compute limitation of the RF-simulated platform rather than a scheduling oversight. The MAC-level per-slot quota assignment continued to track the configured value exactly throughout. We therefore run the agent experiments (Section~\ref{sec:eval-agents}) at offered loads within the two-UE fair/directional region validated here.

\textbf{Bidirectional path check.} Independently of the quota actuator, we verified before every run that the UE-facing PDU session carries traffic in both directions using a dedicated echo probe, distinct from the one-way measurements used for results. This check caught a routing artifact in the RF-simulated core in which uplink traffic addressed to the PDU gateway was silently hairpinned rather than delivered, while downlink remained unaffected; the fix is orthogonal to the MAC quota patch and is reported for reproducibility.

\subsection{Uncoordinated agents destabilize what each protects alone}\label{sec:eval-agents}

\begin{figure}[t]
\centering
\includegraphics[width=1\linewidth]{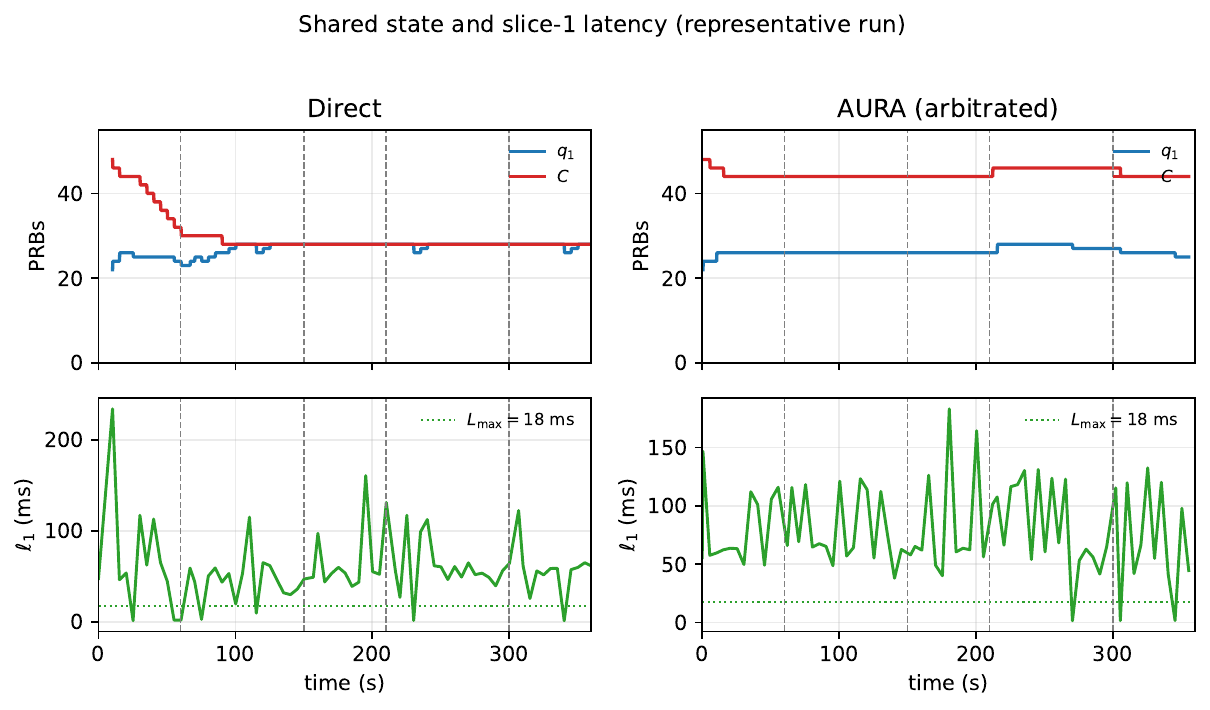}
\caption{Shared-state trajectories and measured slice-1 one-way delay under \textsc{Direct} (left) and \textsc{\sys{}} (right) for one representative run spanning two load cycles (dashed lines mark the four phase transitions at 60, 150, 210, and 300\,s). Measured $\ell_1$ is noisy under both regimes and frequently exceeds $L_{\max}$, reflecting real queueing jitter on the RF-simulated platform; the regime-distinguishing signal is in the top panels, where $C$ continues to move under \textsc{Direct} in the second cycle while \textsc{\sys{}} keeps $q_1$ essentially frozen and damps $C$.}
\label{fig:money}
\end{figure}

Table~\ref{tab:amp} summarizes shared-state amplitude across the four regimes (mean over 5 repetitions). We report amplitude three ways: over the full 360\,s run, after excluding an initial 40\,s settling period common to every regime (the shared initial condition), and restricted to the second load cycle ($t \ge 210$\,s) to test whether the pathology recurs rather than appearing only as a one-time transient. \textsc{Static} shows, by construction, zero movement throughout. In \textsc{Single}, only \Asl{} is active: $q_1$ moves substantially (settled amplitude 19.0 PRBs) while $C$ is untouched under every measure, and this holds throughout the full two-cycle run. This isolates the pathology from either agent's individual behavior: \Asl{} alone does not destabilize the shared cap. In \textsc{Direct}, with both agents active and no arbiter, the settled $C$ amplitude is 8.4 PRBs, and critically this is not a single early transient: restricted to the second load cycle, $C$ still moves by 2.0 PRBs on average, and in 3 of the 5 repetitions it moves by at least 2 PRBs during that second cycle alone (per-repetition values 0, 2, 0, 2, 6). The comparison between \textsc{Single} and \textsc{Direct} uses identical \Asl{} logic and identical traffic, with the only difference being whether \Aen{} is also active. This comparison, together with the recurrence across load cycles, is our central empirical claim: the destabilizing behavior is a property of the \emph{interaction} between agents and recurs under repeated disturbance, exactly as Prop.~\ref{prop:cycle} predicts for opposing best-response maps sharing state under delay, rather than a one-off artifact of the initial condition.

\begin{table}[t]
\caption{Shared-state amplitude (PRBs, mean over 5 reps; lower is better), computed over the full run, after a 40\,s settling period, and restricted to the second load cycle ($t\ge210$\,s).}
\label{tab:amp}
\centering
\begin{tabular}{lccccc}
\toprule
Regime & \multicolumn{2}{c}{$q_1$ amp} & \multicolumn{3}{c}{$C$ amp} \\
 & full & settled & full & settled & cycle 2 \\
\midrule
\textsc{Static} & 0.0  & 0.0  & 0.0  & 0.0 & 0.0 \\
\textsc{Single} & 26.0 & 19.0 & 0.0  & 0.0 & 0.0 \\
\textsc{Direct} & 11.8 & 10.4 & 18.8 & 8.4 & 2.0 \\
\textsc{\sys{}} & 4.4  & 1.4  & 4.0  & 0.4 & 0.4 \\
\bottomrule
\end{tabular}
\end{table}
\begin{table}[t]
\caption{SLA violation rate (\% of measurement windows; lower is better) per slice and regime, mean over 5 runs, using measured one-way delay and measured receiver throughput. A window violates slice~1 if $\ell_1 > L_{\max}=18$\,ms and slice~2 if $r_2 < R_{\min}=8$\,Mbps.}
\label{tab:sla}
\centering
\begin{tabular}{lcccc}
\toprule
 & \textsc{Static} & \textsc{Single} & \textsc{Direct} & \textsc{\sys{}} \\
\midrule
Slice 1, all phases   & 82.8 & 81.7 & 84.5 & 92.9 \\
Slice 1, step\_up only & 82.2 & 80.6 & 85.6 & 92.2 \\
Slice 2, all phases   & 0.6  & 55.0 & 40.2 & 0.3 \\
Slice 2, step\_up only & 0.0  & 54.4 & 46.7 & 0.6 \\
\bottomrule
\end{tabular}
\end{table}

Table~\ref{tab:sla} reports a result we did not expect when we replaced the latency proxy with a real probe, and we report it plainly rather than reshaping the narrative around it. Slice-1 violation rates are high, 80 to 93\%, across every regime, because the RF-simulated platform's baseline one-way delay is itself noisy and frequently exceeds the 18\,ms target regardless of which agents are running (Fig.~\ref{fig:money}); this is a property of the shared-host software radio, not of the interaction pathology, and the absolute level should not be over-read. The signal that does distinguish the regimes cleanly is slice-2 throughput protection: \textsc{Single} and \textsc{Direct} let $q_1$ growth or $C$ shrinkage squeeze the coexisting slice, driving 40 to 55\% throughput violations, while \textsc{\sys{}} holds slice-2 violations at 0.3\%, statistically indistinguishable from \textsc{Static}'s 0.6\%. \textsc{\sys{}} does not lower slice-1's own violation rate; it is numerically the highest of the four regimes (92.9\%, std 3.9\%, versus \textsc{Static}'s 82.8\%, std 5.9\%). Notably, \sys{}'s mean settled $q_1 \approx 25.5$ exceeds \textsc{Static}'s fixed 22, so the higher violation rate cannot be attributed to lower quota alone. To probe the mechanism, we ran a static control at \sys{}'s settled operating point $(q_1,q_2,C)=(26,18,44)$ with no agents: the mean slice-1 violation was 87.0\%~(±2.9), inside \textsc{Static}'s $\pm1\sigma$ band and outside \sys{}'s. This result rules out the cap-scaling mechanism (at this operating point $\sum_i q_i = C$, so the MAC pro-rata shrink path never fires) as the primary driver, but leaves a residual gap of roughly 6 percentage points between the static control and the live \sys{} runs that we attribute to live-dynamics effects such as periodic enforcement churn; with $n{=}5$ repetitions the two distributions overlap in their tails, and we leave precise attribution to future work. We ground the ``automation can jointly harm'' claim in the recurring $C$-excursion of Table~\ref{tab:amp} and in the slice-2 starvation of Table~\ref{tab:sla}, both of which are absent or negligible under \textsc{Static} and under \sys{}.

\subsection{Arbitration removes the cross-slice pathology, at an explicit cost}

Under \textsc{\sys{}} with identical agents and traffic, $q_1$'s settled amplitude falls from 10.4 (\textsc{Direct}) to 1.4, and $C$'s settled amplitude falls from 8.4 to 0.4, with the second-cycle recurrence essentially eliminated (2.0 down to 0.4; Table~\ref{tab:amp}). The arbiter is doing real work, not decoration: pooled over the 5 arbitrated repetitions it rejects 548 proposals, of which 537 (98\%) are rejected by the feasibility invariant $\sum_i q_i \le C$ and 11 (2\%) by the dwell-time check, split 6 on $q_1$ and 5 on $C$ (Fig.~\ref{fig:conflicts}). The deadband check fires zero times in this configuration by construction, since the agents' fixed step size ($\delta=2$) never falls below the configured deadband ($\epsilon=1$); this is an artifact of our parameter choice rather than evidence that the mechanism is inactive, and we note it rather than obscure it. The invariant check dominating the rejection mix is itself informative: it means the efficiency agent's attempts to shrink $C$ below what \Asl{}'s current quota requires are the primary and persistent source of conflict, exactly the structural opposition modeled in Section~\ref{sec:theory}, and the arbiter resolves essentially every instance of it in favor of feasibility.

The honest cost of this configuration is visible directly in Table~\ref{tab:sla} and in Fig.~\ref{fig:money} (right): \sys{} does not reduce slice-1's own latency-violation rate, and numerically it is the worst of the four regimes on that single metric. The arbiter does admit \Asl{}'s quota-growth proposals when capacity allows (settled $q_1 \approx 25.5$ under \sys{}, above \textsc{Static}'s fixed 22), yet the violation rate remains higher; a static control run at \sys{}'s settled operating point narrows but does not close the gap, leaving a live-dynamics contribution as the remaining candidate. The remedy is a proposal that co-ordinates growth of both $q_1$ and $C$ in a single atomic action; our single-variable proposal schema cannot express such coupling, and we identify composite proposals as the natural next extension of the arbiter interface. The trade the current design makes is explicit and is exactly what Prop.~\ref{prop:conv} promises: feasibility, safety-floor compliance, and bounded shared state at every step, not an improvement to every agent's individual objective. In this deployment, that trade buys near-complete protection of the coexisting slice (Table~\ref{tab:sla}) and an order-of-magnitude reduction in recurring shared-state movement (Table~\ref{tab:amp}), in exchange for leaving one already-degraded metric unimproved.

\begin{figure}[t]
\centering
\includegraphics[width=1\linewidth]{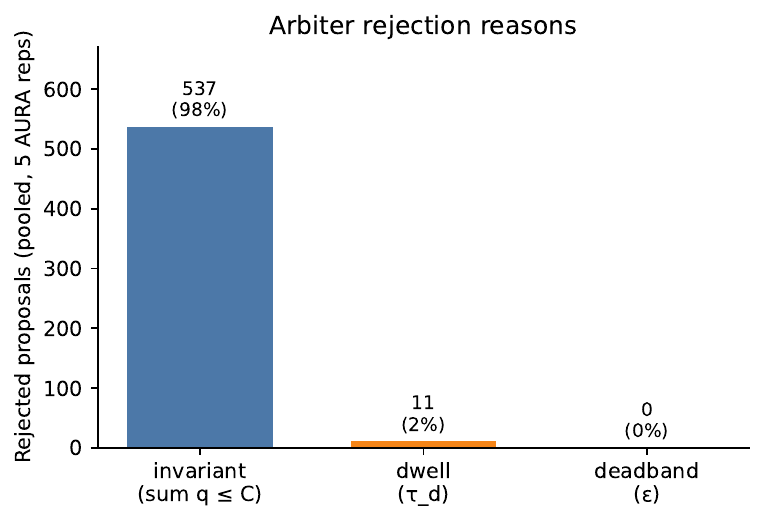}
\caption{Arbiter rejection reasons pooled over 5 arbitrated repetitions: 537 (98\%) invariant ($\sum_i q_i \le C$), 11 (2\%) dwell (6 on $q_1$, 5 on $C$), 0 deadband (step size $\delta=2$ never falls below $\epsilon=1$ in this configuration).}
\label{fig:conflicts}
\end{figure}

\subsection{Control-loop latency budget}\label{sec:latency-budget}

\begin{figure}[t]
\centering
\includegraphics[width=1\linewidth]{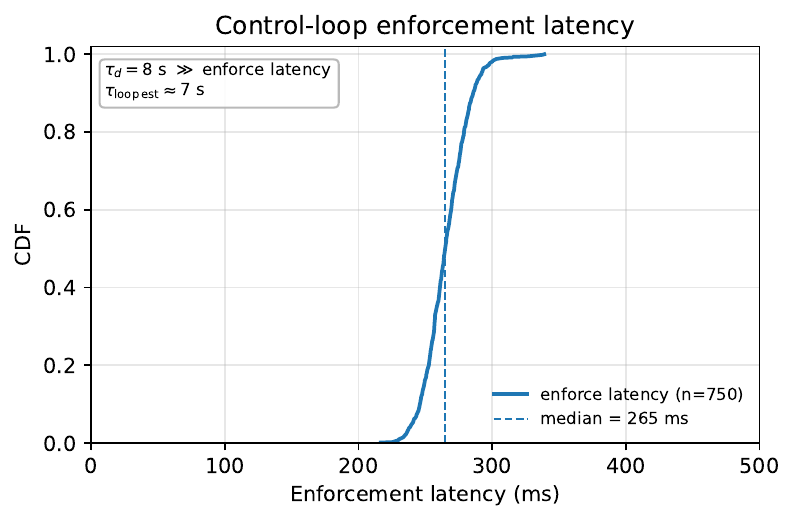}
\caption{Enforcement-latency CDF ($n{=}750$ \texttt{cp\_ms} samples pooled across all regimes); median $\approx$265\,ms. Arbiter dwell $\tau_d{=}8$\,s is set above $\tau_{\mathrm{loop\,est}}{\approx}7$\,s so the condition of Prop.~\ref{prop:conv} holds.}
\label{fig:latency}
\end{figure}

The dominant, directly measured component of $\tau_{\mathrm{loop}}$ is enforcement delivery: the file-driven enforcement step, which writes the quota table into the running gNB pod and triggers its periodic reload, has median latency 264.7\,ms (mean 265.7\,ms, minimum 216.6\,ms, maximum 338.9\,ms, $n{=}750$ pooled across all regimes and both load cycles). Combined with the agent period $T_a=5$\,s and an estimated sampling and observation overhead of ${\sim}1$\,s, this gives an end-to-end loop-delay estimate $\tau_{\mathrm{loop}} \approx 7$\,s, which is the basis for setting the dwell time to $\tau_d = 8$\,s $> \tau_{\mathrm{loop}}$ in \textsc{\sys{}} per Prop.~\ref{prop:conv}. The experiment harness refuses to start an arbitrated run whose configured $\tau_d$ does not exceed this estimate, making the condition of Prop.~\ref{prop:conv} an enforced precondition of the experiment rather than an incidental parameter choice. Arbiter admission logic itself (the three checks of Section~\ref{sec:model}) executes in-process and is negligible relative to the roughly 265\,ms enforcement step. Two conclusions follow: the arbiter's own computational overhead is negligible relative to the loop delay it governs, and in our file-driven deployment the enforcement path, not agent reasoning or arbiter logic, is the bottleneck that determines $\tau_{\mathrm{loop}}$ and therefore the safe dwell-time floor; a production E2-based enforcement path (Section~\ref{sec:impl}) would need to re-measure this component, as it may differ substantially from a file-reload mechanism.

\subsection{Sensitivity}

Our reported matrix fixes $T_a=5$\,s across all regimes and repetitions; sweeping $T_a$ and the step sizes $(\delta,\delta_c)$ to test the amplitude scaling predicted by Prop.~\ref{prop:cycle} is left to future work, as is agent-population scale. We do not claim platform-scale concurrent-agent counts.

\section{Discussion and Future Work}\label{sec:discussion}
 
\textbf{Composite proposals.}
\sys{} currently admits actions as independent single-variable proposals.
As demonstrated in Section~\ref{sec:eval-agents}, this atomic approach successfully bounds shared
state and eliminates cross-slice starvation, but it cannot simultaneously grow a slice
quota~($q_1$) and the energy cap~($C$) in a single step.
The natural extension is \emph{composite proposals}: multi-variable atomic actions that
allow an agent to request additional quota strictly contingent on a simultaneous cap
expansion, enabling Pareto-improving state transitions while preserving the arbiter's
feasibility invariants. Specifically, a composite proposal would allow $A_{\mathrm{sl}}$ to request a simultaneous increase in both $q_1$ and $C$. This guarantees that the newly requested quota is physically backed by the expanded cap, bypassing the invariant conflicts that currently bound slice-1's latency improvements. Since the arbiter evaluates post-state feasibility, this requires expanding the agents' action space, not the underlying arbitration architecture.
 
\textbf{Principled safety--performance trade-offs.}
\sys{} guarantees feasibility and bounded shared state, not per-agent optimality.
Our measurements confirm this precisely: slice-2 throughput violations fall from 40--55\%
under \textsc{Direct} and \textsc{Single} to 0.3\% under \sys{}, while slice-1's
latency-violation rate does not improve.
\textsc{Direct}'s nominally lower slice-1 violation (84.5\% versus \sys{}'s 92.9\%) is
not a genuine gain; it is achieved by permitting \Asl{} to grow $q_1$ at the direct
expense of slice-2, whose throughput is starved by 40--55\%.
No single-slice SLA metric captures this cross-slice cost.
\sys{} makes the true system-wide cost explicit and enforces a feasible operating point
rather than allowing one agent's metric to improve at another slice's expense.
A static control run at \sys{}'s settled operating point $(q_1,q_2,C){=}(26,18,44)$
confirms that the gap to \textsc{Static} on slice-1 latency shrinks to 4.2 percentage
points but is not eliminated, with a residual attributed to live-dynamics effects not
present in the static control; precise attribution is left to future work.
 
\textbf{Scope and scalability.}
Our empirical demonstration uses one cell, two structurally opposed agents, and RF-simulated
radio to isolate the control-plane instability from channel fading and multi-cell mobility.
This scope is deliberate: the pathology is rooted in delayed feedback over shared state,
none of its ingredients require real radio propagation, and the RF simulator provides exact
one-way delay measurement unavailable in over-the-air experiments.
Future work will extend the evaluation to multi UE, multi-cell, many-agent deployments over
an over-the-air (OTA) testbed; the arbiter's population-agnostic admission logic requires no
modification, but quantifying how conflict rates scale with agent population remains an open
question. Beyond scaling the deployment, we plan to replace our rule-based agents with a real
learned policy, either a DRL controller or an LLM-backed planner fine-tuned for the SLA and
efficiency objectives, and validate \sys{} against it on that OTA testbed. This would test the
claim of Section~\ref{sec:model} directly: that the arbiter's guarantees hold regardless of how
an agent decides, not only for the rule-based agents evaluated here.
Enforcement beyond quota semantics (positional PRB partitions, per-slice schedulers) and
energy actuation beyond the PRB-cap proxy are also natural extensions.
 
\textbf{Generality of the framework.}
The propositions in Section~\ref{sec:theory} require only opposing monotone responses over shared
variables and a delayed feedback loop; they are not specific to network slicing.
Any scenario in which autonomous controllers with contradictory objectives share network
state will exhibit the same instability template, including power control versus interference
management, or traffic steering versus energy saving.
The \sys{} arbiter carries over to these domains without modification.
An immediate operational extension, supported by the existing event log, is the continuous
verification of declared effects against observed outcomes, enabling the orchestration layer
to quarantine agents whose actual behavior diverges from their proposals.
\section{Conclusion}

We demonstrated on a live O-RAN stack with real delay and throughput measurements that individually correct agents jointly induce recurring, destabilizing shared-state excursions. To resolve this interaction pathology, we presented a lightweight, three-check arbitration layer carrying a formal convergence guarantee. Our arbiter reduces shared-state movement by an order of magnitude and virtually eliminates cross-slice throughput starvation. While it leaves the protected slice's latency compliance unimproved a trade off we explicitly report as a useful negative result it successfully enforces a safe, feasible operating point.

The broader architectural takeaway is clear: the agentic RAN requires an enforcement boundary where actions meet the network. Crucially, because the arbiter evaluates proposed state modifications rather than internal agent logic, its admission overhead remains negligible regardless of the agent population size. To make both the pathology and its repair reproducible, our enforcement patch, arbiter, and measurement harness will be released.
\appendix
\section{Host CPU on pinned cores}\label{app:cpu}

Table~\ref{tab:cpu} reports mean CPU utilization on the core-pinned softmodem sets during the agent matrix, sampled from the host \texttt{/proc/stat} every 5\,s into the experiment event log. Means stay near 10\% across all regimes and phases, with 95th percentiles under 13.4\% on the gNB set, well below the dual-UE starvation band (64--71\%) documented in Section~\ref{sec:actuator-gate}. The largest regime difference is roughly 2 percentage points (between \textsc{Static} and \textsc{\sys{}} on the gNB set), confirming that host scheduling contention does not explain the shared-state or SLA contrasts in Section~\ref{sec:eval}.

\begin{table}[ht]
\caption{Mean CPU utilization (\%) on pinned core sets, pooled over 5 reps per regime. gNB: cores 0--7; UE$_1$: 8--11; UE$_2$: 12--15.}
\label{tab:cpu}
\centering
\footnotesize
\begin{tabular}{llccc}
\toprule
Regime & Phase & gNB & UE$_1$ & UE$_2$ \\
\midrule
\multirow{3}{*}{\textsc{Static}}
 & steady    & 10.2 & 9.1 & 9.7 \\
 & step\_up  &  9.8 & 9.4 & 9.7 \\
 & step\_down& 9.5 & 8.9 & 10.0 \\
\midrule
\multirow{3}{*}{\textsc{Single}}
 & steady    & 10.4 & 9.1 & 9.1 \\
 & step\_up  & 10.3 & 9.1 & 9.1 \\
 & step\_down& 9.7  & 9.4 & 8.9 \\
\midrule
\multirow{3}{*}{\textsc{Direct}}
 & steady    & 10.6 & 9.8 & 9.4 \\
 & step\_up  & 10.4 & 9.2 & 8.9 \\
 & step\_down& 10.1 & 9.5 & 8.7 \\
\midrule
\multirow{3}{*}{\textsc{\sys{}}}
 & steady    & 11.7 & 9.2 & 8.1 \\
 & step\_up  & 11.7 & 9.2 & 8.1 \\
 & step\_down& 11.1 & 9.1 & 7.8 \\
\bottomrule
\end{tabular}
\end{table}



\begin{thebibliography}{99}\itemsep 2pt

\bibitem{polese2023understanding}
M.~Polese, L.~Bonati, S.~D'Oro, S.~Basagni, and T.~Melodia, 
"Understanding O-RAN: Architecture, interfaces, algorithms, security, 
and research challenges," IEEE Commun. Surveys Tuts., vol.~25, no.~2, 
pp.~1376--1411, 2023.

\bibitem{delrio2025pacifista}
P.~Brach del Prever, S.~D'Oro, L.~Bonati, M.~Polese, M.~Tsampazi, H.~Cheng, and T.~Melodia, ``PACIFISTA: Conflict evaluation and management in Open RAN,'' \emph{IEEE Trans.\ Mobile Comput.}, 2025, early access, doi:\,10.1109/TMC.2025.3570632.

\bibitem{zolghadr2025conflict}
A.~Zolghadr, J.~F.~Santos, A.~Nolan, and L.~A.~DaSilva, ``Mitigating xApp conflicts for efficient network slicing in 6G O-RAN: A graph convolutional-based attention network approach,'' arXiv:2504.17590, 2025.

\bibitem{adamczyk2023conflict}
C.~Adamczyk and A.~Kliks, ``Conflict mitigation framework and conflict detection in O-RAN Near-RT RIC,'' \emph{IEEE Commun.\ Mag.}, vol.~61, no.~12, pp.~12--18, Dec.\ 2023, doi:\,10.1109/MCOM.018.2200752.

\bibitem{foukas2017orion}
X.~Foukas, M.~K.~Marina, and K.~Kontovasilis, ``Orion: RAN slicing for a flexible and cost-effective multi-service mobile network architecture,'' in \emph{Proc.\ ACM MobiCom}, 2017, pp.~127--140.

\bibitem{ksentini2017slicing}
A.~Ksentini and N.~Nikaein, ``Toward enforcing network slicing on RAN: Flexibility and resources abstraction,'' \emph{IEEE Commun.\ Mag.}, vol.~55, no.~6, pp.~102--108, 2017.

\bibitem{kokku2012nvs}
R.~Kokku, R.~Mahindra, H.~Zhang, and S.~Rangarajan, ``NVS: A substrate for virtualizing wireless resources in cellular networks,'' \emph{IEEE/ACM Trans.\ Netw.}, vol.~20, no.~5, pp.~1333--1346, 2012.

\bibitem{foukas2016flexran}
X.~Foukas, N.~Nikaein, M.~M.~Kassem, M.~K.~Marina, and K.~Kontovasilis, ``FlexRAN: A flexible and programmable platform for software-defined radio access networks,'' in \emph{Proc.\ ACM CoNEXT}, 2016, pp.~427--441.

\bibitem{schmidt2021flexric}
R.~Schmidt, M.~Irazabal, and N.~Nikaein, ``FlexRIC: An SDK for next-generation SD-RANs,'' in \emph{Proc.\ ACM CoNEXT}, 2021, pp.~411--425.

\bibitem{cheng2024oranslice}
H.~Cheng, S.~D'Oro, R.~Gangula, S.~Velumani, D.~Villa, L.~Bonati, M.~Polese, G.~Arrobo, C.~Maciocco, and T.~Melodia, ``ORANSlice: An open-source 5G network slicing platform for O-RAN,'' in \emph{Proc.\ ACM MobiCom Workshops (Open AI RAN)}, 2024, doi:\,10.1145/3636534.3701544.

\bibitem{xslice2025}
P.~Yan, J.~Lu, H.~Zeng, and Y.~T.~Hou, ``xSlice: Near-real-time resource 
slicing for QoS optimization in 5G O-RAN using deep reinforcement learning,'' 
arXiv:2509.14343, 2025.

\bibitem{oranwg3}
O-RAN Alliance WG3, ``Near-Real-Time RAN Intelligent Controller Architecture,'' Technical Specification O-RAN.WG3.RICARCH, 2023.

\bibitem{monderer1996potential}
D.~Monderer and L.~S.~Shapley, ``Potential games,'' \emph{Games and Economic Behavior}, vol.~14, no.~1, pp.~124--143, 1996.

\bibitem{liberzon2003switching}
D.~Liberzon, \emph{Switching in Systems and Control}. Boston, MA, USA: Birkh\"auser, 2003.

\bibitem{wu2024netllm}
D.~Wu, X.~Wang, Y.~Qiao, Z.~Wang, J.~Jiang, S.~Cui, and F.~Wang, ``NetLLM: Adapting large language models for networking,'' in \emph{Proc.\ ACM SIGCOMM}, 2024.

\end{thebibliography}
\end{document}